\documentclass{article}
\usepackage[utf8]{inputenc}
\usepackage[english]{babel}
\usepackage{csquotes}
\usepackage{biblatex}
\usepackage{hyperref}

\usepackage{dirtytalk}

\usepackage{ dsfont }
\usepackage{amssymb}
\usepackage{amsmath}        
\usepackage{amsfonts}       
\usepackage{amsthm}         
\usepackage{bbding}         
\usepackage{bm}             
\usepackage{graphicx}       
\usepackage{fancyvrb}       
\usepackage{indentfirst}    
\usepackage{icomma}         
\usepackage{dcolumn}        
\usepackage{booktabs}       
\usepackage{paralist}       
\usepackage{xcolor}         
\usepackage{ dsfont }

\input xy                         
\xyoption{all}
\usepackage{amsmath,amscd}
\usepackage{textcomp}
\theoremstyle{plain}
\newtheorem{veta}{Věta}
\newtheorem{Thm}[veta]{Theorem}
\newtheorem{Prop}[veta]{Proposition}
\newtheorem{Ex}[veta]{Example}

\newtheorem{Lemma}[veta]{Lemma}

\theoremstyle{plain}
\newtheorem{Def}[veta]{Definition}

\theoremstyle{remark}

\theoremstyle{plain}

\newcommand{\C}{\mathbb{C}}
\newcommand{\F}{\mathbb{F}}

\newcommand{\Z}{\mathbb{Z}}

\newcommand{\p}{\mathbb{Z}_{p^\infty}}

\newenvironment{dukaz}{
  \par\medskip \noindent
  \textit{Proof}.
}{

\rightline{$\qedsymbol$}
}

\usepackage{blindtext}
\title{Free and projective LCD codes}
\author{Dominik Krasula
}

\begin{document}
 \maketitle 
\begin{abstract}
LCD codes over finite commutative chain rings are known to be free. However, if the ring is not indecomposable, there always exists an LCD ideal that is projective but not free.  We prove that all two-sided LCD codes are projective. 

We generalise the characterisation of finite commutative Frobenius rings by means of the size condition of a code and its dual to the noncommutative setting. This result is aplied to prove that group codes are LCD if and only if they are generated by a central idempotent. 
\end{abstract}
\textbf{Keywords:}  LCD codes; Frobenius rings; size condition

\medskip 

\noindent \textbf{Mathematics Subject Classification:} 94B05, 16L60; Secondary: 16P10 

\medskip 

\noindent \textbf{Author:} RNDr. Dominik Krasula

Charles University, Faculty of Mathematics and Physics, 

Department of Algebra

Sokolovská 49/83, 186 75 Praha 8, Czech Republic

krasula@karlin.mff.cuni.cz

ORCID: 0000-0002-1021-7364

\medskip  

\noindent  This research was part of the GA UK, project number 310226. It was supported by  GA ČR 26-22734S and  UNCE/24/SCI/022.

\medskip  

\noindent \textbf{Declarations of interest}: I have nothing to declare.

\medskip 

\section{Introduction}

J. L. Massey introduced \textit{linear codes with a complementary dual}, or \textit{LCD codes} for short,  as codes over fields that have zero intersection with their annihilators [Mas. 92]. This was motivated by the observation that they mimic the property of real vector spaces that the annihilator of a subspace is its orthogonal complement.  Over fields of positive characteristic, such as finite fields, there always exists a code that has a nontrivial intersection with its annihilator (Ex. \ref{ExFreeNotLCD}).

LCD codes provide an optimal linear coding solution for the two-user binary adder channel [Mas. 92, § 4] and play a role in counter-measures to side-channel attacks [CG 16]. For further applications, see the literature surveyed in [Dur. 20, Intro.].  Constructions focusing on LCD codes that are
also group codes, i.e. ideals in Frobenius group rings, were studied in [DGKR 22]. [SHSS 19] ignited the study of double circulant LCD codes over Galois rings. 

\smallskip 

Over a general finite ring, the dual of an LCD code need not be its complement. The notion of having a \textit{complement} coincides in module theory with that of \textit{projective modules}, i.e. direct summands of free modules. A projective code, or even a free code, is not necessarily an LCD code (Ex. \ref{ExFreeNotLCD}). Over non-hereditary domains, LCD codes that are not projective exist (Thm. \ref{ThmDomain}). 

The matter is more akin to the case of codes over fields if the base ring is Frobenius. As codes over rings have risen in prominence in recent decades,
several authors have studied the question of whether all LCD codes are free [Dur. 20,  Prop. 4.1], [BBBFM 20, Thm. 2], and [CELNP 25, Thm. 3.15].  We fully solve the commutative case (Thm. \ref{ThmCommMain}).

However, the hypothesis that LCD codes are free is too strong. The articles cited above focused on local rings, over which all projective modules are free. If the ring is not indecomposable, such as a non-local commutative Frobenius ring, there always exists a projective LCD code that is not free (Ex. \ref{ExProjLCD}). 

Frobenius rings are singled out due to their prominence in coding theory. Chain rings and group algebras over fields are Frobenius rings. Furthermore, the duality between left and right ideals makes the study of LCD ideals a natural question about the structure of such rings.

\subsection{Structure of article and main results}

The basic properties of LCD codes and projective codes are gathered in Section \ref{SecPrel}. We generalise known results concerning the direct product of codes (Prop. \ref{PropProd}) and their behaviour under module-theoretic direct sums (Ex. \ref{ExDirSum}). Section \ref{SecSemiperfect} recalls the general version of the Chinese Remainder Theorem (Thm. \ref{ThmCRT1}) and thus reduces the study of LCD codes to indecomposable rings (Thm. \ref{ThmCRTLCD}).

Section \ref{SecTwo-Sided} shows that over \textit{pseudo-Frobenius} (PF) rings, which include QF rings, all two-sided LCD codes are projective (Thm. \ref{ThmCommMain}).

In Section \ref{SecSize}, we generalise the size condition for codes over finite Frobenius rings [Dou. 26] in Theorem \ref{ThmSize}. This condition is used in Section \ref{SecFrob} to show that an LCD code is a direct summand in the category of abelian groups (Lemma \ref{LemmaAbel}).  Using our observation about idempotents in Frobenius rings (Prop. \ref{Propcentral}), this implies that right LCD group codes over finite Frobenius rings are generated by central idempotents (Thm. \ref{ThmLCDmain}). We then extend the result to ideals in Artin algebras that are Frobenius (Thm. \ref{ThmArtin}).

Section \ref{Sectrivial} surveys LCD codes over non-Frobenius rings. It is shown that nonzero codes with a zero annihilator exist if and only if the ring is not Kasch (Thm. \ref{ThmKasch}), and that all LCD codes over a domain are projective if and only if it is a hereditary domain (Thm. \ref{ThmDomain}).

\subsection{Codes over rings}

Throughout this text, by a \textit{code of length $n$ over a ring $R$}, we always mean an $R$-linear code, i.e. a right or left submodule of a free module $R^n$ for some positive integer $n$. Right codes of length one coincide with right ideals. In cases when $C$ is an $(R\text-R)$-subbimodule, we occasionally call it \textit{a two-sided code} to stress this fact. 

While codes over rings have appeared throughout the history of coding theory, the beginning of their modern study is usually associated with [HKCSS 94], where it was realised that several prominent non-linear codes (Nordstrom-Robinson, Kerdock, Preparata, Goethals, and Delsarte-Goethals) \say{\textit{can be very simply constructed as binary images under the Gray map of linear codes over $\Z_4$}}. From a modern perspective, we can view cyclic and constacyclic codes over fields as codes over rings.

Finite Frobenius rings \say{\textit{were singled out as most appropriate for code theoretical purposes}} when J. A. Wood showed that the MacWilliams extension theorem [Wood 99, Thm. 6.3] and the MacWilliams identities [Wood 99, Sec. 8] hold over them.

Two monographs on codes over rings were published in 2017.  The monograph  [SAS 17] by M. Shi, A. Alahmadi, and P. Solé focuses both on \textit{theory and practice} and includes material on noncommutative rings. S. T. Dougherty's work [Dou. 17] treats coding theory (over commutative Frobenius rings) \say{\textit{as a branch of pure mathematics, serving as its own motivation for study}} and focusing on connections with other areas of mathematics. 

Practical motivations also led to the study of LCD codes over non-Frobenius rings, as in [CK 24]. The ring of $p$-adic integers, while infinite, has very workable representations of elements and has been studied by coding theorists since the '90s [CS 95]. Recently, LCD codes over them were described [DS 26].

\section{Preliminaries and basic properties}\label{SecPrel}

This section defines key notions used throughout the text and gathers basic properties of LCD codes. In Section \ref{SecAnn}, we define annihilator maps, characterise dual rings (Thm. \ref{ThmHerbera}), and show that LCD codes behave well with respect to the direct product (Prop. \ref{PropProd}). Nontrivial codes of all lengths $n>1$ are shown to exist (Ex. \ref{Extrivial}). Later, we characterise finite rings where there exist LCD codes of length 1 that are not $0$ or the whole $R$ (Thm. \ref{ThmLCDmain}).

Section \ref{Secproj} introduces projective and free modules and shows that a code is free if and only if its dual is (Prop. \ref{PropFree}). LCD ideals generated by idempotents are discussed (Ex. \ref{ExProjLCD}). As a consequence, we observe that LCD codes do not behave well with respect to module-theoretic direct sums and direct summands. 

As we allow possibly noncommutative rings, we discuss in Section \ref{SecSemiperfect} how the Chinese remainder theorem carries over to this setting (Thm. \ref{ThmCRT1}) and use it to reduce the study of LCD codes to codes over indecomposable rings (Thm. \ref{ThmCRTLCD}).

\subsection{Annihilators and LCD codes}\label{SecAnn}

Let $R$ be a ring and $n$ a positive integer. Then we can represent an element $x\in R^n$ as an $n$-tuple $(x_1,\dots, x_n)$ where each $x_i$ is an element of $R$. For two such elements $x, y\in R^n$, we define their product as their \textit{dot product}, i.e. $xy:=\sum_{i=1}^n x_iy_i$. 

Other products on $R^n$ are studied in the literature. A classic example is the Hermitian inner product over $\C$.   We consider the dot product as the most natural due to its connection with $n$-ary $R$-linear forms (Lemma \ref{Lemmadual}).

Throughout this text, we consider \textit{annihilators} with respect to this notion of product. In other words, if $X\subseteq R^n$ is a set, we define the \emph{left and right annihilators of $X$}  as \[{}^\perp X:=\{s\in R^n\mid sX=0 \}\quad and \quad X^\perp:=\{s\in R^n\mid Xs=0\}.\] 

It is easy to see that the left (right) annihilator of $X$ is a left (right) submodule of $R^n$, and it coincides with the left (right) annihilator of the right (left) submodule of $R^n$ generated by $X$. The annihilator of a two-sided code is a two-sided code. Over commutative rings, all codes are two-sided.

A left (right) code $C\leq R^n$ is \textit{a left (right) LCD code} if $C\cap C^\perp=0$  ($C\cap {}^\perp C=0$). If a two-sided code is both a left and right LCD code, we call it a \textit{two-sided LCD code}.

\begin{Ex}\label{Extrivial}
    Let $R$ be a ring. Then free LCD codes of all lengths exist:

    In particular, $0$ is an LCD code and given any positive integer $n$, the cyclic free codes generated by the canonical basis elements $\epsilon_i$ for any $i\leq n$ are free LCD codes.
\end{Ex}

Annihilator maps are order-reversing maps between the posets $\mathcal{L}({}_RR^n)$ and $\mathcal{L}(R^n_R)$. If they are bijections, they are necessarily lattice anti-isomorphisms, also called \textit{lattice duality}. 

A ring where annihilator maps induce duality between $\mathcal{L}({}_RR)$ and $\mathcal{L}(R_R)$ is sometimes called an \textit{annihilator ring} in the literature, as this property is equivalent to every one-sided ideal being an annihilator. We refer to such rings as \textit{dual rings}, as in [HN 85], motivated by the following theorem. The first part is [AHM 00, Thm. 2.1], the second then follows from [HN 85, Thm. 4.2].
\begin{Thm}\label{ThmHerbera}
Let $R$ be a ring such that the lattices of left and right ideals are anti-isomorphic.

Then the annihilator maps induce a duality between left and right ideals, and the ring $R$ is a semiperfect ring with essential socles. 

Furthermore, $R$ has a Nakayama permutation and its socles coincide. 
\end{Thm}
 It seems to be an open question whether, over any dual ring $R$, there is a lattice self-duality between left and right submodules of $R^n$ for any $n>1$. But it is true if the ring is PF  [Lam 99, Prop. 19.46]. 

\medskip 

Over a general ring $R$,  the annihilator of an intersection of two right codes need not equal the sum of their respective annihilators. But for any two left codes $C, D\leq R^n$,it holds that $(C+D)^\perp=C^\perp \cap D^\perp$.   
\begin{Prop}\label{PropProd}
    Let $R$ be a ring, $n, m$ positive integers, and $C\leq R^n$ and $D\leq R^m$ left codes.

    Then $(C, D)\leq R^{n+m}$ is an LCD code if and only if $C$ and $D$ are LCD codes.
\end{Prop}
\begin{dukaz}
    As a left $R$-module, $(C,D)$ decomposes as $(C, 0^m)\oplus (0^n, D)$. Hence 
    \[(C, D)^\perp =(C, 0^m)^\perp \cap(0^n, D)^\perp = (C^\perp, R^m)\cap (R^n, D^\perp)=(C^\perp, D^\perp),  \]
    showing that $(C, D)\cap (C^\perp, D^\perp)$ is zero if and only if $C\cap C^\perp = 0 =D\cap D^\perp$.
\end{dukaz}

\subsection{Projective and free codes}\label{Secproj}

Let $R$ be a ring and $C\leq R^n$ a right code. We say that $C$ is \textit{free} if it is isomorphic, as a right $R$-module, to $R^s$ for some positive integer $s$. Over finite rings, a free code is, up to isomorphism, determined by its rank $s$, i.e. if $C\cong R^s\cong R^t$, then $s=t$. A ring satisfying this property is said to have \textit{invariant basis number (IBN)}. Not all rings have IBN, but semiperfect rings, such as QF rings, do. 

 Free codes of finite length are thus analogues of finite-dimensional vector spaces, as they have a basis and are determined by the cardinality of the basis. All modules over a ring are free if and only if the ring is a division ring. However, there exist rings over which all codes are free, yet non-free modules exist, e.g. principal ideal domains  that are not fields (Prop. \ref{Propfir}).

 If a ring has a positive characteristic, there always exists a free code that is not LCD:
\begin{Ex}\label{ExFreeNotLCD}
    Let $R$ be a ring of characteristic $q>0$.

    Then the right $R$-submodule of $R^q$ generated by $1^q=(1, \dots, 1)$ is a free code of rank 1. It is not an LCD code as it is fully contained in its left annihilator. 
\end{Ex}

A right $R$-module is said to be \textit{projective} if it is a direct summand of some free right $R$-module.  Over a local ring, all projective modules are free.  The dual of the notion of a projective module is an injective module. QF rings are characterised as rings in which projectives and injectives coincide. Over non-local rings, projective LCD codes that are not free exist.

\begin{Ex}\label{ExProjLCD}
    Let $R$ be a ring and let $1=e+f$ be a decomposition of unity into the sum of two orthogonal idempotents. Then $Re$ and $Rf$ are left projective modules, since $Re\oplus Rf=R$.

   Since $(Re)^\perp=(1-e)R$, the right ideal $Re$ is an LCD code if and only if $(1-e)Re=0$.

    In particular, if $R$ can be decomposed as a product of two rings $S\times T=R$, then $S\times 0:=\{(s,0)\mid s\in S\}$ is a projective two-sided LCD ideal in $R$. If $R$ has IBN, then $S\times 0$ is not free, since it is a proper direct summand of a free module of rank 1.
\end{Ex}
For a QF ring, the condition $(1-e)Re=0$ implies that the idempotent is central (Prop. \ref{Propcentral}). The relation between central idempotents and the block decomposition of a ring is described in Theorem \ref{ThmCRT1}.

While the class of projective modules is closed under direct summands and direct sums, the same is not true for LCD codes. Example \ref{ExProjLCD} shows that if $R$ is indecomposable but not local, then there exists a direct summand that is not an LCD code. The following gives an example of two LCD codes whose direct sum is not an LCD code.

\begin{Ex}\label{ExDirSum}
Let $R:=\F_3$.

Let then $C:=(0,1,1)R$ and $D:=(1,1,0)R$ are LCD codes, but $(1,2,1)=(1,-1,1)$ is a nonzero element of $(C\oplus D)\cap (C\oplus D)^\perp$.    
\end{Ex}

We end this section by characterising free codes in terms of their duals. The case for local rings was solved in [CELNP 25, Thm. 3.11].
\begin{Prop} \label{PropFree} Let $R$ be a QF ring and $C\leq R^n$ a left code.

(1) If $C$ is free of rank $s$, then $C^\perp$ is free of rank $n-s$.  

(2) If $C$ is projective, then $C^\perp$ is a projective right module.
\end{Prop}
By symmetry, if a code has a free dual, then the code itself is free.
\begin{dukaz}
(1) If $D\leq R^n$ is a left code such that $D\cong C$, then $C^\perp\cong D^\perp$ as right $R$-modules [Lam 99, Thm. 15.21].  Over QF rings, free modules are injective, so if $C$ is free, it is a direct summand of $R^n$. Thus we can find $s\leq n$ and an isomorphism $C\cong R^s\times 0^{n-s}$. The dual of the latter module is $0^s\times R^{n-s}$, which is a free module isomorphic to $C^\perp$. 

(2) By Lemma \ref{Lemmadual}, $C^\perp$ is isomorphic to the dual $(R^n/C)^*=Hom_R(R^n/C, R)$. Because $C$ is projective, it is a direct summand of $R^n$, and hence $R^n/C$ is isomorphic to its necessarily projective complement. Over QF rings, $Hom_R(-, R)$ is a Morita duality, so in particular, the dual of a projective module is injective, so $C^\perp$ is injective. Since $R$ is QF, injective modules are projective. 
\end{dukaz}

\subsection{Semiperfect rings and Chinese remainder theorem}\label{SecSemiperfect}

Semiperfect rings generalise finite rings, as they are characterised by the property that unity decomposes as a sum of finitely many pairwise orthogonal local idempotents [Mül. 70, Thm. 1].  Over such a ring, primitive and local idempotents coincide. Local rings and QF rings are semiperfect. We refer the reader to [Lam 91, Chaps. 7, 8] for more on idempotents and semiperfect rings.

For commutative finite rings, a strong version of the Chinese remainder theorem (CRT) is known and used to study codes over such rings [Dou. 17, Cor. 2.1]. We now give its general noncommutative version [AF 92, §7] with several well-known corollaries and a stronger version for Artinian rings [Iov. 16, Cor. 2.3], which was generalised to semiperfect Noetherian rings in [Kra. 26, Thm. 9] and used to characterise rings where MacWilliams conditions hold [Iov. 22, Thm. 2.1]. If a ring is left or right Noetherian, then MacWilliams extension conditions for left ideals imply that the ring is QF [GS 22, Cor. 2.10]. Semiperfect rings satisfy condition (1) [AF 92, Thm. 7.9], so the theorem applies. 
\begin{Thm}\label{ThmCRT1}
    Let $R$ be a ring. The following are equivalent:

(1)  There exists a decomposition of unity $1=e_1+\dots + e_s$ into a sum of centrally primitive pairwise orthogonal idempotents.  

 (2)   There exist rings $B_1, \dots, B_s$, called \emph{blocks of $R$}, each indecomposable in the category of rings, such that $R$ is isomorphic to their direct product. 

  If (2) holds, then this \emph{block decomposition} is unique up to reordering and isomorphism. In particular, there is a permutation $\sigma\in S_s$ such that $B_{\sigma(i)}\cong e_iRe_i$ as rings. 

Furthermore: 

(i) $R$ is commutative if and only if each $B_i$ is a local commutative ring.

(ii) $R$ is left Noetherian/left Artinian if and only if each $B_i$ is left Noetherian/left Artinian. Furthermore, over an indecomposable semiperfect left and right Noetherian ring, all simple modules have the same characteristic and are either all finite or have the same infinite cardinality.

(iii)  If $R$ is a left  Artinian ring, then it is a left  MacWilliams ring
if and only if $R$ decomposes as a product of a finite Frobenius ring and a QF ring that has no non-trivial finite modules. Moreover, in this case, $R$ is also a right Artinian and right MacWilliams ring.

(iv) $R$ is PF/QF/Frobenius if and only if each $B_i$ is. 

(v) (Wedderburn-Artin theorem) $R$ is semisimple if and only if each $B_i$ is a ring of square matrices over a division ring.
\end{Thm}

Given a product of rings $R:=\prod^s_{i=1} R_i$, we can represent elements of $R$ as $s$-tuples, such that for any $r\in R$ the $i$-th coordinate $r_i$ is an element of $R_i$. This provides a natural ring isomorphism $CRT\colon \prod^s_{i=1} R_i\to R$ that induces a map $CRT\colon \prod^s_{i=1} R^n_i\to R^n$. As in the case of self-dual codes [DHS 99, Thm. 2.1], LCD codes behave well with respect to CRT. The following can be seen as a generalisation of [LL 15, Thm. 4.2]  used to characterise LCD codes over principal ideal rings.
\begin{Thm}\label{ThmCRTLCD} Let $R_1, \dots, R_s$ be rings, and for each $i\leq s$, let $C_i$ be a code of length $n$ over $R_i$.

    Then $C:=CRT(C_1, C_2, \dots, C_s)$ is a code over $R:=R_1\times \dots \times R_s$  of length $n$ that is projective  (free) if and only if each $C_i$ is.
    
   Further, $C$ is an LCD code if and only if each $C_i$ is an LCD code over $R_i$. 
\end{Thm}
\begin{dukaz}
    The first claim is standard. To prove the \textit{Further} part, observe that for each $i\leq s$, an element $x\in R_i^{n_i}$ is an annihilator of $C_i$ if and only if $\epsilon _ix$ is an annihilator of $C$. 
\end{dukaz}

\section{Two-sided LCD codes are projective}\label{SecTwo-Sided}

We show that two-sided LCD codes over rings with duality between left and right submodules are projective. This class of rings includes QF rings. Note that over a QF ring, a code is its own double annihilator. It follows that a right code is an LCD code if and only if its annihilator is a left LCD code.


\medskip 

Let $R$ be a commutative ring and $C\leq R^n$ be a left LCD code over $R$.
Proofs that $C$ is free are based on the observation that $C\oplus C^\perp=C+ C^\perp$, since the intersection $C\cap C^\perp$ is assumed to be zero.

Working over finite commutative local Frobenius rings, the proof of [BBBFM 20, Thm. 2]
uses a size condition (Thm. \ref{ThmSize}) to ensure that $C\oplus C^\perp=R^n$, showing that $C$ is projective, and hence free, as projective modules over a local ring are free.
This method was generalised to arbitrary commutative local Frobenius rings in [CELNP 25, Thm. 3.15] via composition length.

Our approach is based on the proof of [Dur. 20, Prop. 4.1], which shows that LCD codes over finite commutative chain rings are free. This result was specialised for rings of the form $\Z_{p^e}$ in [DS 26, Lemma 3.2]. The key idea is to deduce that $C\oplus C^\perp=R^n$ not from the properties of modules, but by utilising the fact that the annihilator induces a duality of lattices, not only an anti-isomorphism of posets. This allows us to extend the result beyond the Artinian setting.

\begin{Thm}\label{ThmCommMain}
    Let $R$ be a ring and $n$ a positive integer such that there exists a lattice duality $l\colon \mathcal{L}(R^n_R)\to  \mathcal{L}({}_RR^n)$, and let $r$ be its inverse.   Further, let $C\leq R^n$ be a two-sided code such that $l(C)\cap C=0=C\cap r(C)$.
    
     Then $C$ is projective as both a right and a left $R$-module.

     In particular, if $R$ is a local ring, then $C$ is free.
\end{Thm}
\begin{dukaz}
    Applying the lattice duality $r$ to the assumption $C\cap l(C)=0$
    yields \[r(C)+rl(C)=R^n.\]
Since $l$ and $r$ are mutually inverse, $C=rl(C)$, so
\[r(C)+C=R^n.\]
Because $r(C)\cap C$ is assumed to be zero, it follows that $r(C)\oplus C=R^n$. Since both $r(C)$ and $C$ are $(R\text-R)$-bimodules, $C$ is a direct summand of $R^n$ both in the category of left and the category of right $ R$-modules, and hence projective by definition.

The \textit{in particular} part follows since projective modules over a local ring are free [AF 92, Cor. 26.7].
\end{dukaz}

 A commutative finite ring is a product of local commutative rings (Thm. \ref{ThmCRT1}), so using Theorem \ref{ThmCRTLCD}, we can restrict our attention to local rings. Over local rings, all LCD codes are free.  So, to characterise LCD codes, it is enough to characterise when a free code is LCD.  The following result is [BBBFM 20, Cor. 2]. The case for chain rings was solved earlier in [LL 15, Thm. 3.5, Cor. 3.6].  
\begin{Prop}\label{ThmLCDlocal}
    Let $R$ be a finite local commutative Frobenius ring. Let $C\leq R^n$ be a free code of rank $s$ with generator matrix $G$.

    Then $C$ is an LCD code if and only if the $s\times s$ matrix $GG^T$ is invertible.
\end{Prop}

\section{Size condition for noncommutative Frobenius rings}\label{SecSize}

We give a common generalisation of two known characterisations of finite Frobenius rings by means of the sizes of codes and their duals (Thm. \ref{ThmSize}). This is achieved by combining an extension of a classical result by T. Nakayama  from ideals to codes (Lemma \ref{LemmaNakPerm}) with our recent result on simple ideals paired by a Nakayama permutation [Kra. 26, Prop. 11].

This \textit{size condition} is a key step in the proof of Theorem \ref{ThmLCDmain}, but we consider it important enough to treat it as an independent result.

T. Honold proved that a finite ring is Frobenius if and only if  $|I||I^\perp|=|R|$ for any left ideal $I\leq R$ [Hon. 01, Thms. 1(v), 2]. It is well known that, over commutative rings, this characterisation extends to codes. This result is commonly attributed to Wood's foundational paper [Wood 99, Sec. 9, Appendix A]. For example, [BBBFM 20, Proof of Thm. 2] cites [Wood 99] without further explanation. But it seems that it was never explicitly proved until recently in [Dou. 26].

\medskip 

Initially, when generalising (quasi-)Frobenius algebras to Artinian rings, Nakayama defined them as Artinian rings with a Nakayama permutation [Nak. 41]. 

\begin{Def}\label{DefQF}
    Let $R$ be a semiperfect ring, and let $e_1, \dots, e_m$ be a complete set of pairwise orthogonal primitive idempotents. Furthermore, assume that they are ordered in such a way that for some $k\leq m$ it holds that if $k<i\leq m$, then there exists $j\leq k$ such that $e_iR\cong e_jR$, or equivalently $Re_i\cong Re_j$.  Let $\pi$ be a permutation on $\{1,2,\dots, k\}$. 
    \begin{itemize}
    
\item For $i\leq k$, we consider a simple right module $V_i:=top(e_iR)=e_iR/(e_iJ(R))$ and a simple left module $V'_i:=top(Re_i)=Re_i/(J(R)e_i)$. We define  \emph{multiplicity}, denoted by $\mu_i$, as the maximal number of isomorphic copies of $V_i$ in the indecomposable decomposition of $top(R)$ as a right module.

   \item   $\pi$ is a \emph{Nakayama permutation } if   \[
    \forall i \leq k \colon \quad  soc(e_iR) \cong top(e_{\pi(i)}R) \quad and \quad  soc(Re_{\pi(i)}) \cong top(Re_{i}). \]

In such a case, we say that $V_i$ and $V'_{\pi^{-1}(i)}$ are \emph{paired by the Nakayama permutation}.

                 \item  A ring is \emph{pseudo-Frobenius}, or PF for short,  if it is a linearly compact ring with a Nakayama permutation.

       \item  A ring is \emph{quasi-Frobenius}, or QF for short,  if it is an Artinian ring with a Nakayama permutation.

   \item  A ring is \emph{Frobenius} if it is a QF ring whose Nakayama permutation $\pi$ \emph{preserves multiplicities}, i.e. $\mu_i=\mu_{\pi(i)}$.
        \end{itemize}
\end{Def}
The proof of the following is inspired by the proof of [Nak. 41, Thm. 7]. The \textit{left version} of this result works with essentially the same proof. 

\begin{Lemma}\label{LemmaNakPerm}
Let $R$ be a  QF ring with a Nakayama permutation $\pi\in S_k$. Further, let $C\leq D\leq R^n$ be two right codes such that $D/C\cong V_i$ for some $i\leq k$.

    Then ${}^\perp C/{}^\perp D \cong V'_{\pi^{-1}(i)}$.
\end{Lemma}
\begin{dukaz}
    First, observe that for any $x\in R^n$, the first isomorphism theorem applied on the right $R$-homomorphisms $d\mapsto xd$ gives 
    \[         D/(D\cap x^\perp)\cong xD.\]
      
   Now if $x\in {}^\perp C$ then $Rx\subseteq {}^\perp C$ and hence \[C\subseteq ({}^\perp C)^\perp\subseteq (Rx)^\perp=x^\perp,\]
   showing $C\subseteq D\cap x^\perp$. Since $C$ is maximal in $D$, this means that the factor $D/(D\cap x^\perp)\cong xD$ is either zero if $D\cap x^\perp=D$ or a simple right ideal isomorphic to $V_i$ if $D\cap x^\perp=C$.
   
   Since this holds for any $x\in {}^\perp C$, it follows that the two-sided ideal $({}^\perp C)D$ is a direct sum of simple right ideals, so it is semisimple as a right ideal and hence contained in the right socle. Over QF rings, left and right socles coincide (Thm. \ref{ThmHerbera}), so $({}^\perp C)D$ is also contained in the right socle.  More precisely, we find a set of indices $I\subseteq \{1, \dots, m\}$ such that
   \[({}^\perp C)D= \bigoplus_{j\in I} soc(e_jR)\]
and we already established that each simple direct summand is isomorphic to $V_i\cong soc(e_iR)$. It then follows from the definition of Nakayama permutation that multiplication by $e_{\pi^{-1}(j)}$ on the right, for any $j\in I$, does not annihilate the ideal $({}^\perp C)D$.  Thus, as a left ideal, it is a finite direct sum of isomorphic copies of $V'_{\pi^{-1}(i)}$, possibly zero.

\medskip 
   
   By the assumption,  we have $D/C\cong V_i\cong e_itop(R)$ so $e_i\notin Ann_R(D/C)=Ann_R(e_itop(R))$. Since $e_itop(R)e_i$ is a nonzero corner ring of $top(R)$, it follows that $(D/C)e_i=(De_i+C)/C\neq 0$.    
   In particular, $De_i$ is not entirely contained in $C$. 
   
   So let $b\in De_i\setminus C$. This means that $b=be_i$. Using the first isomorphism theorem as in the beginning of the proof, we have the following
   \[({}^\perp C)b\cong ({}^\perp C)/({}^\perp C\cap {}^\perp b), \]
   where \[{}^\perp C\cap {}_{R^n}Ann(b)={}^\perp C\cap {}^\perp(bR)={}^\perp(C+bR).\] 
 Because  $D/C$ is simple, we then have $bR+C=D$, so 
   \[({}^\perp C)b\cong ({}^\perp C)/({}^\perp C\cap {}^\perp b)= {}^\perp C/{}^\perp D.\]

    Because $b = be_i$, we have $({}^\perp C)b = ({}^\perp C)be_i \subseteq Re_i$.  Recalling that $({}^\perp C)b\subseteq ({}^\perp C)D \subseteq  soc(R)$, we see that $({}^\perp C)b$ is a semisimple left, so it is contained in $soc(R)e_i=soc(Re_i)$.

By the definition of a Nakayama permutation, $soc(Re_i)\cong V'_{\pi^{-1}(i)}$. So $({}^\perp C)b\cong {}^\perp C/{}^\perp D$ is either isomorphic to  $V'_{\pi^{-1}(i)}$ or zero. But over a QF ring, the anihilator maps are lattice dualities, so ${}^\perp C/{}^\perp D$ is nonzero, and the conclusion follows.
\end{dukaz}
We now characterise the sizes of simple modules. 
\begin{Lemma}\label{LemmaSize}
    Let $R$ be a finite ring with a Nakayama permutation $\pi\in S_k$ and let $i\leq k$.

    Then $|V_i|=|V'_i|=|End_R(V_i)|^{\mu_i}$
\end{Lemma}
\begin{dukaz}
  Each simple right $R$-module $V_i$ is also a simple right $top(R)=R/J(R)$-module. So it is enough to prove the claim for semisimple rings. We have a decomposition
\[top(R)\cong \prod_{i=1} ^k M_{\mu_i}(End(V_i))\]
in the category of rings. If $S$ is a matrix ring over a division ring, then its (unique up to isomorphism) simple right and simple left modules have the size $|End_R(V_i)|^{\mu_i}$ as they correspond to row and column ideals, respectively.
\end{dukaz}
\begin{Thm}\label{ThmSize}
    Let $R$ be a finite ring.

    Then $R$ is Frobenius if and only if for any positive integer $n$ and any left code $C\leq R^n$ it holds that $|C||C^\perp|=|R^n|$.
\end{Thm}
\begin{dukaz}
  For $n=1$, the condition implies Honold's size condition, and hence the ring is Frobenius [Hon. 01, Thms. 1(v), 2].

  Now assume that $R$ is Frobenius. We prove that $C$ and $R^n/C^\perp$ have the same size. Consider the composition series of $C$:
  \[0=C_0\leq C_1\leq \dots \leq C_\lambda=C.\]
Because $R$ is QF, the annihilators give a composition series of a module $R^n/C^\perp$:
\[C^\perp =C_{\lambda}^\perp \leq C^\perp_{\lambda -1 }\leq \dots \leq C_0^\perp =R^n.\]

It follows that $|C|=\prod_{1\leq i\leq \lambda} |C_i/C_{i-1}|$ and $|R^n/C^\perp|=\prod_{1\leq i\leq \lambda}|C^\perp_{i-1}/C^\perp_{i}|$. By Lemma \ref{LemmaNakPerm}, for any positive $i\leq \lambda$, simple modules $C_i/C_{i-1}$ and $C^\perp_{i-1}/C^\perp_{i}$ are paired by a Nakayama permutation. So by [Kra. 26, Prop. 11] that have isomorphic endomorphism rings. In particular, their endomorphism rings have the same size, so $C_i/C_{i-1}$ and $C^\perp_{i-1}/C^\perp_{i}$ have the same size if and only if they have the same multiplicity (Lemma \ref{LemmaSize}), which is ensured, since $R$ is Frobenius, so Nakayama's permutation preserves multiplicities.
\end{dukaz}

\section{LCD ideals in Frobenius rings}\label{SecFrob}

In this section, we prove that left LCD ideals over Frobenius Artin algebras are generated by a central idempotent (Thm. \ref{ThmLCDmain}), and thus, they are also right LCD ideals.

We have the following necessary condition for a code to be LCD.
\begin{Lemma}\label{LemmaAbel}
    Let $R$ be a finite Frobenius ring and $C\leq R^n$ a left LCD code.
    
    Then $(C, +)$ is a direct summand of $(R^n, +)$ in the category of abelian groups.       
\end{Lemma}
\begin{dukaz}
     Since $C\cap C^\perp=0$, the set $C+C^\perp$ equals the direct sum $C\oplus C^\perp$ in the category of abelian groups. Since $|C\oplus C^\perp|=|C||C^\perp|=|R^n|$ (Thm. \ref{ThmSize}), $C\oplus C^\perp$ is a subgroup of $R^n$ that has the same size as the finite group $R^n$, so it follows that $C\oplus C^\perp=R^n$ as abelian groups. 
\end{dukaz}
In general, we cannot conclude that $C+C^\perp$ is a left code, but it holds if $n=1$. 

Before proceeding, we need to establish a key property of idempotents in QF rings (Prop. \ref{Propcentral}).    It was recently observed [Kra. 25, Prop. 16] that semicentral idempotents in a PF ring are necessarily central. Hence, a PF ring cannot be nontrivially represented as a ring of upper triangular matrices.  This is a key observation needed for the proof of Theorem \ref{ThmLCDmain}. We present a proof of this property, due to J. Žemlička, which, unlike the original proof, is independent of the notion of formal matrix ring representation. First, we need the following lemma.

\begin{Lemma}\label{LemmaHom} Let $R$ be a semiperfect ring with a Nakayama permutation $\pi\in S_k$ and essential socles, and let $e\in R$ be an idempotent. 
\begin{enumerate}
\item[(1)] If $Hom(eR,soc((1-e)R))\ne 0$, then $Hom((1-e)R,soc(eR))\ne 0$,
\item[(2)] if $Hom(eR,(1-e)R)\ne 0$, then $Hom((1-e)R,eR)\ne 0$.
 \end{enumerate} 
\end{Lemma}
\begin{proof} 
Denote
\[
I=\{i\le k\mid e_iR\hookrightarrow eR\},\ \ J=\{j\le k\mid e_jR\hookrightarrow (1-e)R\}
\]
and observe that $I\cup J=\{1,\dots k\}$. 

(1) By the hypothesis, there exist indices $i\in I$ and $j\in J$ such that $Hom(e_iR,soc(e_jR))\ne 0$, which means that $e_iR/e_iJ(R)\cong soc(e_jR)$ and so $\pi(i)=j$. It is now easy to see that there is $r< k$ for which $\pi^r(i)\in J$ and $\pi^{r+1}(i)=\pi^r(j)\in I$. Hence $Hom(e_{\pi^r(i)}R,soc(e_{\pi^r(j)}R))\ne 0$ and so $Hom((1-e)R,soc(eR))\ne 0$.

(2) Assume that $Hom(eR,(1-e)R)\ne 0$ and $Hom((1-e)R,eR)= 0$. Then $Hom((1-e)R,soc(eR))= 0$, hence $Hom(eR,soc((1-e)R))=0$ by (1).

Since $Hom(eR,(1-e)R)\ne 0$, there exists a non-zero homomorphism $\varphi\colon e_iR \to e_jR$ for some $i\in I$, $j\in J$. Let $X:=\varphi^{-1}(soc(e_jR))\le e_iR$.  Note that $\varphi(X)$ equals $soc(e_jR)$ because it is a simple submodule of $e_jR$. Since it is essential, it is fully contained in any nonzero submodule, hence it is in particular contained in the image of $\varphi$.

Since $\pi$ is a Nakayama permutation and $e_{\pi^{-1}(j)}R$ is projective, there exists a homomorphism 
$\psi\colon e_{\pi^{-1}(j)}R\to X$ such that $\varphi\psi(e_{\pi^{-1}(j)}R)=soc(e_jR)$. 

Then $\varphi\psi\in Hom(e_{\pi^{-1}(j)}R,soc(e_jR))$, which implies that $\pi^{-1}(j)\in J$ and $Hom(e_{\pi^{-1}(j)}R, e_iR)\ne 0$. Now, since $\pi^{-1}(j)\in J$ and $i\in I$, we get that $Hom((1-e)R,eR)\ne 0$, a contradiction.
\end{proof}
\begin{Prop}\label{Propcentral}
    Let $R$ be a ring with a Nakayama permutation and essential socles, and let $e\in R$ be an idempotent.

    Then $eR(1-e)=0$ if and only if $(1-e)Re=0$.
\end{Prop}
\begin{dukaz}
    By symmetry, it is enough to prove one implication. Let us assume $(1-e)Re$ is nonzero.
    
    First, recall that in any ring,  there is an isomorphism of abelian groups   \[Hom_R(eR, (1-e)R)\cong (1-e)Re.\]
    In particular, it follows that if $(1-e)Re$ is nonzero, then so is $Hom_R(eR, (1-e)R)$. By Lemma \ref{LemmaHom}(2), it then follows that $Hom_R((1-e)R, eR)$ is nonzero. Now there is an isomorphism of abelian groups 
\[Hom_R((1-e)R, eR)\cong eR(1-e),\]
so it follows that $eR(1-e)$ is also nonzero. 
\end{dukaz}
We can now prove the main Theorem of this section.
\begin{Thm}\label{ThmLCDmain}
    Let $R$ be a finite Frobenius ring.

    If $C\leq R$ is a left LCD ideal, then $C$ is generated by a central idempotent.

 In particular,  $R$ is indecomposable if and only if $0$ and $R$ are the only left LCD ideals.
\end{Thm}
\begin{dukaz}
By Lemma \ref{LemmaAbel}, $R=C\oplus C^\perp$ as abelian groups. In particular, we find $e\in C$ and $f\in C^\perp$ such that $1=e+f$.  We show that $e$ and $f$ are orthogonal idempotents: \[e=1e=(e+f)e=e^2+fe\]
now since $e\in C$ and $C$ is left ideal, we see that $fe\in C$. Similarly, since $f\in C^\perp$, we have that $fe\in C^\perp$. Thus  $fe\in C\cap {}^\perp C=0$, showing that $e=e^2$. Since $e$ is an idempotent, it follows that $f=1-e$ is also idempotent and they are orthogonal.

Since $C$ is a left ideal, and $e\in C$ we see that $Re\leq C$ and similarly $(1-e)R\leq C^\perp$ where $f=1-e$. It follows that $(1-e)Re$ is contained in $C\cap C^\perp$ and hence $(1-e)Re=0$, so $e$ is a central idempotent by Prop. \ref{Propcentral}.

Since $e$ is central, we have $eR=Re\subseteq C$. Now take $x\in C$, then \[x=1x=(e+f)x=ex+fx.\] 
It follows that $fx\in C\cap C^\perp=0$ we see that $x=ex$ so $C\subseteq eR$, showing $C=Re=Re$.

The \textit{in particular} part follows, since a ring is indecomposable if and only if $0$ and $1$ are the only central idempotents.
\end{dukaz}

\subsection{Artin algebras}

In the rest of the section, we show that this result extends to Artin algebras. It remains an open question whether every QF ring has all LCD ideals generated by central idempotents.

\begin{Def}
    Let $R$ be a ring and $A$ a commutative Artinian ring.

    We say that $R$ is an \emph{Artin algebra over $A$} if there is a ring homomorphism $\psi \colon A\to R$ such that $Im(\psi)\subseteq Z(R)$ and $R$ is finitely generated as a $\psi(A)$-module.

    Given a (left or right) $R$-module $M$, we use $\lambda_A(M)$ to denote the composition length of $M$ viewed as a module over a commutative ring $A$.
\end{Def}
\textit{Example:}

\noindent      \textbullet~  Finite rings are Artin algebras over their centre.

\noindent      \textbullet~ Given a field $K$, the finite-dimensional $K$-algebras are Artin algebras.

\noindent      \textbullet~  If $R$ is a commutative Artinian ring and $G$ is a finite group, then the group ring $R[G]$ is an artin algebra over $R$. If $R$ is finite Frobenius, $R[G]$ is Frobenius [Wood 99, Ex. 4.4(v)].

    \medskip 

 We note that an Artin algebra is Frobenius if and only if its right socle is isomorphic to its right top [Iov. 16, Thm. 1.7], i.e. unlike in the general case of Artinian rings [Lam 99, Ex. 16.2], it is enough to impose this condition only on one side. This was previously known to hold for finite-dimensional algebras [Nak. 49, Cor. 2] and finite rings [Hon. 01, Thm. 1(ii)].

We now formulate a lemma that allows us to compute the $A$-length of simple modules and, hence, by the Jordan-Hölder theorem, the $A$-length of any finitely generated module over an Artin algebra. 
\begin{Lemma}\label{LemmaA-leng}
    Let $R$ be an Artin algebra over $A$, and let $S$ be a simple right $R$-module of multiplicity $\mu$.

    Then $D:=End_R(S)$ is an Artin algebra over $A$ and $\lambda_A(S)=\mu\cdot \lambda_A(D)$.
\end{Lemma}
\begin{dukaz}
The endomorphism ring of any finitely generated $R$-module is an Artin $A$-algebra [ARS 95, Prop. II.1.1(b)].   

Since we can represent $S$ as a row in a matrix ring $M_\mu(D)$ of a corresponding direct summand of $top(R)$, it decomposes, as an $A$-module, to the direct sum of $A$-modules $\epsilon_iD$ for $i\leq \mu$. That is $M_\mu(D)\cong (D, 0, \dots, 0) \oplus \dots \oplus (0,\dots, 0,D)$. So the formula follows.
\end{dukaz}
This gives rise to a new characterisation of Frobenius Artin algebras. The proof is essentialy the same as the proof of Theorem \ref{ThmSize}.
\begin{Thm}\label{ThmArtin}
    Let $R$ be an Artin algebra over an Artinian ring $A$.

    Then $R$ is Frobenius if and only if for any left code $C\leq R^n$ it holds that $\lambda_A(C)+\lambda_A(C^\perp)=\lambda_A(R^n)$.

    In such a case, if  $C\leq R$ is a left LCD ideal, then $C$ is generated by a central idempotent.
 \end{Thm}
 \begin{dukaz}
The endomorphism rings of simple modules paired by a Nakayama permutation are isomorphic and thus have the same composition length as $A$-modules. In particular, two simple modules paired by a Nakayama permutation have the same composition length as $A$-modules. It follows that the first part is proved in an analogous manner to the proof of Theorem \ref{ThmSize}.

If $R$ is Frobenius and $C$ a left LCD ideal, it follows that $\lambda_A(C\oplus C^\perp)=\lambda_A(R)$. So $C+C^\perp=C\oplus C^\perp=R$ as abelian groups. The conclusion then follows by the same proof as Theorem \ref{ThmLCDmain}.
 \end{dukaz}

\section{LCD codes with trivial annihilators}\label{Sectrivial}

While annihilator maps in QF rings induce a duality between left and right ideals, a general annihilator map can have a nontrivial kernel. Codes whose annihilator are zero are thus a trivial case of LCD codes.  

On the opposite side of the spectrum from QF rings are \textit{domains}, that is, nonzero (possibly noncommutative) rings with no nonzero zero divisors. It then trivially follows that all left ideals are LCD codes. Hence, domains where all right LCD codes are projective coincide with right hereditary domains (Thm. \ref{ThmDomain}). All right LCD ideals are free if and only if all right ideals are free. Such rings were characterised by P. M. Cohn and his work on \textit{firs}, which are necessarily domains (Prop. \ref{Propfir}).

\medskip 
 
 A ring is PF if and only if the standard duality  $^*=Hom_R( -, R)$ is a Morita duality.  Annihilators are closely related to standard duality, even if the ring is not PF.  Recall that for any ring, $(R^n)^*\cong R^n$.

The isomorphism $R^n\cong Hom_R(R^n, R)$ can be realised in the following way: the canonical basis vectors $\epsilon_1, \dots, \epsilon_n$ are a \textit{free basis} of  $R^n$, so any right (left) $R$-homomorphism $\varphi\colon R^n\to R$ is uniquely determined by its images on these vectors. Thus $\varphi$ is the same map as multiplying from the left (right) by a vector $(\varphi(\epsilon_1), \dots, \varphi (\epsilon_n))$, giving the following classical lemma:
\begin{Lemma}\label{Lemmadual}

Let $C\leq R^n$ be a left code. 

Then $C^\perp \cong Hom_R(R^n/C, R)$ as right $R$-modules.    
\end{Lemma}
\begin{dukaz}
   There is a natural $R$-isomorphism between the right $R$-module $Hom_R(R^n/C, R)$ and the right $R$-submodule of $Hom_R(R^n, R)$ consisting of left $R$-homomorphisms whose kernel is $C$.  The isomorphism $R^n\cong (R^n)^*$ then induces an isomorphism with elements of $R^n$ annihilated by $C$.
\end{dukaz}

\begin{Thm}\label{ThmKasch}
    Let $R$ be a ring.

    Then, $R$ is left Kasch if and only if for any positive integer $n$, the total code $R^n$ is the only left code of length $n$ whose annihilator is zero. 
\end{Thm}
\begin{dukaz}
It is classical that if  $R$ is not left Kasch, there exists a proper left ideal whose right annihilator is trivial. Now assume $R$ is left Kasch. Since $C\neq R^n$, we can, by Zorn's Lemma, find a maximal left submodule $M\leq R^n$ containing $C$. Since $M^\perp \leq C^\perp$, if $M$ has a nonzero annihilator, so does $C$. But $M^\perp \cong Hom_R(R^n/M, R)=(R^n/M)^*$ (Lemma \ref{Lemmadual}). Because $R$ is left Kasch, the simple left module $R^n/M$ injects into $R$, so its standard dual is nonzero. 
\end{dukaz}
   \textit{Example:} 
   
\noindent      \textbullet~ Rings with a Nakayama permutation, such as PF rings, are Kasch.

\noindent \textbullet~ Local rings have (up to isomorphism) only one simple left module, so they are left Kasch if and only if their left socle is nonzero.

\noindent \textbullet~ Rings with zero socles cannot be Kasch. This includes all domains that are not division rings.

 \medskip

A ring is called \textit{right hereditary} if all right ideals are projective right modules. By Kaplansky's Theorem, a ring is right hereditary if and only if all right codes are projective. Recalling that over QF rings, projective and injective modules coincide, it follows that a QF ring is right hereditary if and only if it is semisimple. A quiver algebra without oriented cycles over a field is hereditary [ARS 95, III. 1.4]. Hereditary commutative domains are equivalent to Dedekind domains.

P. M. Cohn defined \textit{a right free ideal ring}, or \textit{right fir} for short, as a ring where all right ideals are free of unique rank in [Cohn 64]. The assumption of unique rank serves to \textit{exclude pathologies}. In particular, it implies that right firs are domains, which is mentioned without proof in Cohn. 

\begin{Prop} \label{Propfir} Let $R$ be a ring. 

(1) If $R$ is a right fir, then it is a domain.

(2) $R$ is a left and right fir if and only if it is a (possibly noncommutative) unique factorisation domain.

(3) $R$ is a principal right ideal domain if and only if it is a right Noetherian right fir.  It then follows that all codes are free.
\end{Prop}
\begin{dukaz}
    (1)      Let $a\in R$ be a nonzero element. Then we have a short exact sequence
    \[ 0 \to Ann_R(a) \to R\to aR\to 0. \]
    Now both $aR$ and $Ann_R(a)$ are right ideals, so they are free. In particular, $aR$ is projective, so the sequence splits and hence $R\cong aR \oplus Ann_R(a)$. By assumption, $R$ has a unique rank 1, so $Ann_R(a)$ needs to be a zero ideal.

(2)  [Cohn 64, Thm. 2.8].

(3) The first part is  [Cohn 06, Prop. 2.2.2]. Over a principal right ideal domain, all submodules of a free module are free by  [Lam 99, Cor. 2.27].
\end{dukaz}

By  Wedderburn's little theorem, a domain is finite if and only if it is a field. Codes over the infinite domain of \textit{p}-adic integers were studied utilising their canonical connection to finite chain rings, starting with the paper of A. R. Calderbank and N. J. Sloane [CS 95]. We use $\p$ to denote this ring following [CS 95], but we note that many authors reserve $\p$ to denote the Prüfer group instead.  

MDS codes exist over $\p$ for all lengths, ranks, and primes $p$, and the existence of self-dual codes has been characterised  [DP 06]. Recently, LCD codes over $\p$ were studied [DS 26].

It is a classical result that $\p$ is a principal ideal domain. So all codes, and hence all LCD codes, over $\p$ are free (Prop. \ref{Propfir}). Thus, the assumption of free-ness explicitly stated in [DS 26, Thms. 4.4 \& 4.9] is redundant as it is satisfied automatically. The weaker condition that all LCD ideals are projective characterises right hereditary domains:
\begin{Thm}\label{ThmDomain}
    Let $R$ be a ring.
    
    Then all nonzero right ideals are LCD codes if and only if $R$ is a domain. 
    
    It then follows that the following are equivalent:

     (1) $R$ is a right hereditary ring.

    (2) All right LCD ideals are projective.

    (3) All right codes are projective.
\end{Thm}
\begin{dukaz}
 If $R$ is a domain, then no nonzero element has a nonzero annihilator, so they are LCD codes. For the opposite direction, it suffices to assume that all nonzero principal right ideals have zero left annihilator. But then follows that their generators do not have nonzero annihilators and thus the ring has no nonzero zero divisors. 

 The implication (1)$\implies$(2) holds in any ring. Since we established that over domain all right ideals are LCD ideals, (2)$\implies$(1) by definition. Implication (3)$\implies$(2) is trivial. Implication (1)$\implies$(3) holds in any ring by Kaplansky's theorem. 
\end{dukaz}

\medskip

\textbf{Use of AI tools declaration}: I am the sole author of the presented article.  Grammarly, Writefull and Gemini 3.1 Pro were used to provide suggestions, corrections, and enhancements to content I have authored. I have carefully verified the accuracy, validity, and appropriateness of the final form of the manuscript.

\medskip

\noindent \textbf{Bibliography:}

[AHM 00] P. N. Ánh, D. Herbera, C. Menini. (2000). Baer and Morita Duality. \textit{J. Alg.} 232(2).  462-484. https://doi.org/10.1006/jabr.2000.8377

[AF 92] F. W. Anderson, K. R. Fuller. (1992). Rings and Categories of Modules. Graduate Texts in Mathematics 13. (Springer). https://doi.org/10.1007/978-1-4612-4418-9

[ARS 95] M. Auslander, I. Reiten, S. O. Smal\o. (1995). Representation Theory of Artin Algebras. (Cambridge University Press).
https://doi.org/10.1017/CBO9780511623608

[BBBFM 20] S. Bhowmick, A. Fotue-Tabue, E. Martínez-Moro, R. Bandi, S. Bagchi. (2020). Do non-free LCD codes over finite commutative Frobenius rings exist? \textit{Des. Codes Cryptogr.} 88. 825-840. https://doi.org/10.1007/s10623-019-00713-x

[CG 16] C. Carlet, S. Guilley. (2016). Complementary dual codes for counter-measures to side-channel attacks. \textit{Adv. Math. Commun.} 10(1). 131-150. https://doi.org/10.3934/amc.2016.10.131

[CK 24] H. S. Choi, B. Kim. (2024). Various structures of cyclic codes over the non-Frobenius ring $\mathbb{F}_p[u, v] /\left\langle u^2, v^2, u v, v u\right\rangle$. \textit{Adv. Math. Commun.} 18(2). 304-327. https://doi.org/10.3934/amc.2023030

[CELNP 25] E. Camps-Moreno, C. Espinosa-Valdez, H. H. López, L. Nuñez-Betancourt, Y. Pitones. (2025). Bounds and MacWilliams Identities for Codes over Artinian Rings. https://arxiv.org/abs/2509.05850

[Cohn 64] P. M. Cohn. (1964). Free Ideal Rings. \textit{J. Alg.} 1(1). 47-69. https://doi.org/10.1016/0021-8693(64)90007-9

[Cohn 06] P. M. Cohn. (2006). Free Ideal Rings and Localization in General Rings. (Cambridge University Press). https://doi.org/10.1017/CBO9780511542794

[CS 95] A. R. Calderbank, N. J. A. Sloane. (1995). Modular and p-adic Cyclic Codes. \textit{Des. Codes Cryptogr.} 6(1). 21-35. https://doi.org/10.1007/BF01390768

[DGKR 22] S. T. Dougherty, J. Gildea, A. Korban, A. M. Roberts. (2022). Group LCD and group reversible LCD codes. \textit{Finite Fields Appl.} 83. 102079. https://doi.org/10.1016/j.ffa.2022.102079

[DHS 99] S. T. Dougherty, M. Harada, P. Solé. (1999). Self-Dual Codes over Rings and the Chinese Remainder Theorem. \textit{Hokkaido Math. J.} 28(2). 253-283. https://doi.org/10.14492/hokmj/1351001213

[Dou. 17]  S. T. Dougherty. (2017). Algebraic Coding Theory Over Finite Commutative Rings. SpringerBriefs in Mathematics. (Springer Cham). https://doi.org/10.1007/978-3-319-59806-2

[Dou. 26] S. T. Dougherty. (2026). A characterization of finite commutative Frobenius rings and applications to algebraic coding theory. \textit{Adv. Math. Commun.} 20. 113-117. https://doi.org/10.3934/amc.2025032

[DP 06] S. T. Dougherty, Y. H. Park. (2006). Codes Over the p-adic Integers. \textit{Des. Codes Cryptogr.} 39(1). 65-80. https://doi.org/10.1007/s10623-005-2542-x

[DS 26] S. T. Dougherty, E. Saltürk. (2026). Linear Complementary Dual Codes over the p-adic Integers. \textit{Adv. Math. Commun.} 22. 217-232. https://doi.org/10.3934/amc.2026009

[Dur. 20] Y. Durgun. (2020). On LCD codes over finite chain rings. \textit{Bull. Korean Math. Soc.} 57(1). 37-50. http://bkms.kms.or.kr/journal/view.html?doi=10.4134/BKMS.b181173

[GS 22] P. A. Guil Asensio, A. K. Srivastava. (2022). MacWilliams extending conditions and quasi-Frobenius rings. \textit{J. Alg.} 605. 394-402. https://doi.org/10.1016/j.jalgebra.2022.05.005

[HKCSS 94] A. R. Hammons Jr., P. V. Kumar, A. R. Calderbank, N. J. A. Sloane, P. Solé. (1994). The Z4-linearity of Kerdock, Preparata, Goethals, and related codes. \textit{IEEE Trans. Inform. Theory} 40(2). 301-319. https://doi.org/10.1109/18.312154

[HN 85] C. R. Hajarnavis, N. C. Norton. (1985). On Dual Rings and their Modules. \textit{J. Alg.} 93(2). 253-266. https://doi.org/10.1016/0021-8693(85)90159-0

[Hon. 01] T.  Honold. (2001). Characterization of finite Frobenius rings. \textit{Arch. Math.} 76(6). 406-415. https://doi.org/10.1007/PL00000451

[Iov. 16] M. C. Iovanov. (2016). Frobenius–Artin algebras and infinite linear codes. \textit{J. Pure Appl.
Algebra.}  220(1). 560-576. https://doi.org/10.1016/j.jpaa.2015.05.030

[Iov. 22] M. C. Iovanov. (2022). On Infinite MacWilliams Rings and Minimal Injectivity Conditions. \textit{Proc. Am. Math. Soc.} 150(11). 4575-4586. https://doi.org/10.1090/proc/15929

[Kra. 25] D. Krasula. (2025). Formal matrix representations of pseudo-Frobenius and Frobenius rings. \textit{arXiv preprint}. https://doi.org/10.48550/arXiv.2504.14285

[Kra. 26] D. Krasula. (2026). Endomorphism rings of Simple Modules and Block Decomposition. \textit{J. Algebra Its Appl.} 25(12).  https://doi.org/10.1142/S0219498826501562

[Lam 91] T. Lam. (1991). A First Course in Noncommutative Rings.  Graduate Texts in Mathematics 131. (Springer). https://doi.org/10.1007/978-1-4684-0406-7

[Lam 99]  T. Lam. (1999). Lectures on Modules and Rings.  Graduate Texts in Mathematics 189. (Springer). https://doi.org/10.1007/978-1-4612-0525-8

[LL 15] X. Liu, H. Liu. (2015). LCD codes over finite chain rings. \textit{Finite Fields Appl.} 34. 1-19. https://doi.org/10.1016/j.ffa.2015.01.004

[Mas. 92] J. L. Massey. (1992). Linear codes with complementary duals. \textit{Discrete Math.} 106-107. 337-342. https://doi.org/10.1016/0012-365X(92)90563-U

[Mül. 70] B. J. Müller. (1970). On Semi-perfect Rings.  \textit{Illinois J. Math} 14(3). 464-467. https://doi.org/10.1215/ijm/1256053082

[Nak. 41] T. Nakayama. (1941). On Frobeniusean Algebras II. \textit{Ann. Math.} 42(1). 1-21. https://doi.org/10.2307/1968984

[Nak. 49] T. Nakayama. (1949). Supplementary remarks on Frobeniusean algebras. I. \textit{Proc. Japan Acad.} 25(7). 45-50. https://doi.org/10.3792/pja/1195571908

[SAS 17] M. Shi, A. Alahmadi, P. Solé. (2017). Codes and Rings: Theory and Practice. (Academic Press, Inc.). https://doi.org/10.1016/C2016-0-04429-7

[SHSS 19] M. Shi, D. Huang, L. Sok, P. Solé. (2019). Double Circulant Self-dual and LCD Codes over Galois Rings. \textit{Adv. Math. Commun.} 13(1). 171-183. https://doi.org/10.3934/amc.2019011

[Wood 99] J. A. Wood. (1999). Duality for Modules over Finite Rings and Applications to Coding Theory. \textit{Amer. J. Math.} 121(3). 555-575. https://doi.org/10.1353/ajm.1999.0024
\end{document}